\documentclass[letterpaper, 10 pt, conference]{ieeeconf}
\IEEEoverridecommandlockouts
\usepackage{graphics}
\usepackage{epsfig}
\usepackage{mathptmx}
\usepackage{times}
\usepackage{amsmath}
\usepackage{amssymb}
\usepackage{booktabs}
\usepackage{algorithm}
\usepackage{algpseudocode}
\usepackage{url}
\usepackage{caption}
\usepackage{graphicx}
\usepackage{microtype}
\usepackage[colorlinks=true,linkcolor=blue,citecolor=blue,urlcolor=blue]{hyperref}
\usepackage{bm}
\usepackage{mathtools}
\usepackage{cuted}

\usepackage{amsthm}
\theoremstyle{plain}
\newtheorem{theorem}{Theorem}
\newtheorem{lemma}{Lemma}

\newtheorem{assumption}{Assumption}
\theoremstyle{definition}
\newtheorem{definition}{Definition}
\newtheorem{remark}{Remark}

\usepackage{todonotes}

\newcommand{\R}{\mathbb{R}}
\newcommand{\E}{\mathbb{E}}
\newcommand{\N}{\mathcal{N}}

\newcommand{\Od}{\mathcal{O}}
\newcommand{\dt}{\,\mathrm{d}t}
\newcommand{\dW}{\,\mathrm{d}W}
\newcommand{\norm}[1]{\left\| #1 \right\|}

\newcommand{\inner}[2]{\langle #1,\, #2 \rangle}

\title{\LARGE \bf
Wasserstein Stability of Contracting Flows: \\ 
Effective Rates, Euler Self-Correction, and Noise Tightening
}

\author{Ali Baheri
\thanks{Ali Baheri is with the Department of Mechanical Engineering, Rochester Institute of Technology, Rochester, 14623, USA
        {\tt\small akbeme@rit.edu}}
}

\begin{document}

\maketitle

\begin{abstract}
Contraction theory guaranties exponential convergence between trajectories of a stable 
nonlinear system. When initial conditions are uncertain and represented as probability 
distributions, as in ensemble control, Bayesian estimation, and generative modeling, 
this guaranty extends to the distributional level via Wasserstein distance. However, 
the classical distributional bound is tight only for linear systems; for nonlinear 
dynamics, it can be significantly conservative because it collapses the spatially 
varying local contraction rate to a single worst-case constant, discarding distributional 
information entirely. We address three concrete consequences of this conservatism. 
First, we derive a tighter Wasserstein bound by replacing the worst-case rate with a 
displacement-weighted distributional average of the local contraction rate, which 
strictly improves upon the classical bound for every nonlinear contracting system. 
Second, we provide the first theoretical characterization of the self-correcting Euler 
discretization error under contraction: the error profile is non-monotone, peaks at a 
universal time that depends only on the contraction rate, and then decays exponentially, 
a behavior absent in non-contracting dynamics. Third, we prove that nonlinear 
contracting drifts always achieve strictly smaller stationary variance than a linear 
system sharing the same worst-case contraction rate, formally establishing the 
noise-rejection advantage of nonlinear controllers. All results are validated on a 
representative suite of one- and two-dimensional vector fields.
\end{abstract}

\section{Introduction}\label{sec:intro}
Contraction theory~\cite{lohmiller1998contraction} guaranties that
if a vector field $v(t,x)$ satisfies the one-sided Lipschitz
condition
$\inner{v(t,x)-v(t,y)}{x-y}\le -\lambda\norm{x-y}^2$
for all $t,x,y$, then every pair of trajectories converges at
rate~$e^{-\lambda t}$~\cite{aminzare2014contraction}.
When the initial state is not a single point but a probability
distribution, as in ensemble control~\cite{fornasier2011fluid},
Bayesian state estimation, or generative
modelling~\cite{lipman2023flow}, this trajectory-level guaranty
lifts to the distributional level via optimal
transport~\cite{villani2009ot}. Related constrained generative
constructions modify flow-matching vector fields to enforce logical
specifications, making the stability of the induced transport a practical
concern~\cite{baheri2026logic}. The Wasserstein-2 distance between
any two distributions $\mu_t,\nu_t$ evolving under the same
$\lambda$-contracting field satisfies
\begin{equation}\label{eq:intro_classical}
    W_2(\mu_t,\nu_t) \le e^{-\lambda t}\,W_2(\mu_0,\nu_0).
\end{equation}
Crucially, the derivation of~\eqref{eq:intro_classical} replaces
the \emph{local} contraction rate at each state pair, which, for
nonlinear fields, varies across the state space and can far
exceed~$\lambda$, with the \emph{global worst case}~$\lambda$
before averaging over the coupling.  For a cubic restoring force
$v(x)=-\lambda x -\beta x^3$, the local rate at a pair $(x,y)$
equals $\lambda + \beta(x^2\!+\!xy\!+\!y^2)$, which exceeds
$\lambda$ by an order of magnitude when the distributions are spread
far from the origin.  The bound~\eqref{eq:intro_classical} discards
this distributional information entirely, and as a result, it is
tight only for affine dynamics.
This loss of information has three concrete consequences for the
analysis of nonlinear contracting systems.
\smallskip\noindent\textbf{(G1) Conservative Wasserstein bound.}
By reducing the spatially varying contraction rate to a single
scalar~$\lambda$, the bound~\eqref{eq:intro_classical} can
overestimate the actual Wasserstein distance by a factor that grows
with both the nonlinearity strength and the initial spread of the
distributions.  No existing result recovers the lost distributional
information.
\smallskip\noindent\textbf{(G2) Missing discretisation theory.}
When a contracting flow is implemented via explicit Euler integration,
the truncation forcing decays as trajectories approach equilibrium,
producing an error profile that rises, peaks, and then
falls, a behaviour qualitatively different from the monotonic saturation
predicted by standard ODE error theory.  This self-correcting
behavior is widely observed but has no theoretical characterization:
neither the peak time, the peak magnitude, nor the decay rate has
been identified.
\smallskip\noindent\textbf{(G3) No nonlinear stochastic improvement.}
For a linear system driven by additive noise of intensity~$\sigma$,
the stationary variance is exactly $\sigma^2/(2\lambda)$.  For
nonlinear $\lambda$-contracting drifts, this same expression provides
only an upper bound.  Whether a nonlinear system achieves
\emph{strictly} smaller variance and by how much has remained
open, preventing certification that nonlinear controllers outperform
linear baselines in noise rejection.
\medskip
All three limitations stem from a single source: the worst-case
reduction $\lambda_{\mathrm{loc}}\to\lambda$ that discards the
distributional structure of the contraction rate.
We resolve all three by retaining this structure.

\subsection*{Contributions}

We resolve all three gaps.  The key idea, common to all three results,
is to replace the worst-case rate~$\lambda$ with a
displacement-weighted distributional average of the local rate.

\smallskip\noindent\textbf{C1.\ Effective contraction rate
(Theorem~\ref{thm:eff_method}; G1).}
We prove the refined bound
\begin{equation}\label{eq:intro_eff}
    W_2(\mu_t,\nu_t) \le
    \exp\!\Big(\!-\!\int_0^t\!
    \lambda_{\mathrm{eff}}(s;\gamma_s)\,ds\Big)\,
    W_2(\mu_0,\nu_0),
\end{equation}
where $\lambda_{\mathrm{eff}}(t;\gamma_t)\ge\lambda$ is the
displacement-weighted average of $\lambda_{\mathrm{loc}}$ under the
transported coupling.  Equality holds if and only if $v$ is affine,
so~\eqref{eq:intro_eff} strictly refines~\eqref{eq:intro_classical}
for every nonlinear contracting field.

\smallskip\noindent\textbf{C2.\ Self-correcting Euler envelope
(Theorem~\ref{thm:selfcorr_method}; G2).}
The Euler discretization error satisfies
$e_k \le C_0\,h\,t_k\,e^{-\lambda t_k}$,
implying a universal peak time $t^\star=1/\lambda$ (independent of
$h$, $L$, and $\mu_0$), peak magnitude $\mathcal{O}(h/\lambda)$,
and exponential terminal decay.  The self-correction vanishes for
non-contracting dynamics.

\smallskip\noindent\textbf{C3.\ Strict stochastic variance bound
(Theorem~\ref{thm:sde_method}; G3).}
For $dX=v(X)\dt+\sigma\dW$ we prove
$\E_{\pi_\sigma}\norm{X-x^\star}^2 =
d\sigma^2/(2\lambda_{\mathrm{eff}}^\infty)$
with $\lambda_{\mathrm{eff}}^\infty > \lambda$ strictly for every
non-affine $C^1$ drift: nonlinear contraction always rejects noise
more effectively than the linear baseline.

\medskip\noindent
All results are validated on linear, cubic, rotational, and expanding
fields in 1-D (exact $W_2$) and 2-D (sliced $W_2$).

\smallskip\noindent\textbf{Related work.}
The Wasserstein contractivity bound~\eqref{eq:intro_classical} has
deep roots in the coupling method for Markov
processes~\cite{dobrushin1970,lindvall2002}, and was placed on
rigorous footing via optimal-transport gradient-flow theory by
Ambrosio, Gigli, and Savar\'e~\cite{ambrosio2008gradient}.
Bolley, Gentil, and Guillin~\cite{bolley2013} established $W_2$
contraction for log-concave diffusions, and
Eberle~\cite{eberle2016} extended coupling techniques to non-convex
settings via reflection couplings.  All these results invoke the
\emph{global} one-sided Lipschitz constant~$\lambda$; the
distribution-dependent refinement that we introduce in
Theorem~\ref{thm:eff_method} does not appear in this line of work.

For discretisation under dissipativity, the monograph of Stuart and
Humphries~\cite{stuart1996} analyzes long-time behavior of numerical
schemes for dissipative ODEs; Higham, Mao, and
Stuart~\cite{higham2002strong} provide strong convergence of
Euler--Maruyama under monotone coefficients; and Hutzenthaler, Jentzen,
and Kloeden~\cite{hutzenthaler2011} show that the standard Euler scheme
can diverge for superlinearly growing drifts.  Dalalyan~\cite{dalalyan2017}
and Durmus and Moulines~\cite{durmus2019} derive non-asymptotic
$W_2$ bounds for Langevin discretisation under strong
log-concavity. These works bound the \emph{stationary bias} but do not
identify the non-monotone, self-correcting error profile that
contraction produces (Theorem~\ref{thm:selfcorr_method}).
Complementary work on metriplectic conditional flow matching builds
dissipative structure directly into the learned vector field and a
structure-preserving sampler, with continuous- and discrete-time
stability guarantees~\cite{baheri2025metriplectic}. Here, instead, we
characterize the transient explicit-Euler error profile for a general
contracting drift.

On stationary concentration, the Bakry--\'Emery
framework~\cite{bakry2014} yields Poincar\'e and log-Sobolev
inequalities under curvature bounds, and
Cattiaux and Guillin~\cite{cattiaux2014} connect these to
$W_2$~convergence rates.  The classical bound
$\E_\pi\norm{X-x^\star}^2 \le d\sigma^2/(2\lambda)$ follows
from any of these approaches; our contribution
(Theorem~\ref{thm:sde_method}) is the proof that the inequality
is \emph{strict} for every non-affine contracting drift, together
with the explicit effective-rate identity~\eqref{eq:lambda_eff_inf}
that quantifies the improvement.

\section{Problem Formulation}\label{sec:problem}

\subsection{Two dynamical regimes}\label{subsec:regimes}

\noindent\textbf{Regime~(D): Deterministic flow.}
Consider the initial-value problem
\begin{equation}\label{eq:D_ode}
    \dot x_t = v(t,x_t),\quad t\in[0,T],\quad
    x_0\sim\mu_0\in\mathcal P_2(\R^d),
\end{equation}
where the velocity field $v:[0,T]\times\R^d\to\R^d$ may depend on
time and the initial condition is drawn from a distribution $\mu_0$
with a finite second moment.  The flow map $\phi_{t,0}$ carries each
initial state to its position at time~$t$, and the law evolves as
$\mu_t=(\phi_{t,0})_\#\mu_0$.

When two ensembles start from different distributions $\mu_0$
and $\nu_0$ — for instance, two candidate initializations of a
generative model, or two subpopulations in a multi-agent system
— the distance $W_2(\mu_t,\nu_t)$ between their evolved laws is the natural
measure of distributional stability. Regime (D) is the setting for
Theorems \ref{thm:eff_method} and \ref{thm:selfcorr_method}.

\medskip\noindent\textbf{Regime~(S): Stochastic dynamics.}
The It\^o SDE
\begin{equation}\label{eq:S_sde}
    dX_t = v(X_t)\dt + \sigma\dW_t
\end{equation}
with autonomous drift $v:\R^d\to\R^d$, constant $\sigma>0$, and
standard $d$-dimensional Brownian motion~$W_t$.  Autonomy endows the
system with a unique invariant measure
$\pi_\sigma\in\mathcal{P}_2(\R^d)$ describing the long-run
statistical equilibrium.  Physically, Regime~(S) models any
contracting system subject to persistent perturbations—process noise
in control, stochastic gradients in optimization, or thermal
fluctuations.  The invariant measure $\pi_\sigma$ reflects the balance
between contraction (pulling toward $x^\star$) and diffusion
(dispersing); its dispersion $\E_{\pi_\sigma}\norm{X-x^\star}^2$ is
the object of Theorem~\ref{thm:sde_method}.

\subsection{Wasserstein-2 distance and distributional
contraction}\label{subsec:w2}

We measure distributional convergence with the Wasserstein-2
distance~\cite{villani2009ot}, which is natural in our setting: its
coupling structure meshes directly with trajectory-based contraction
analysis, and it is well-defined between empirical measures, enabling
particle-based validation.

For $\mu,\nu\in\mathcal{P}_2(\R^d)$,
\begin{equation}\label{eq:w2_def}
    W_2(\mu,\nu) \triangleq \Bigl(
    \inf_{\gamma\in\Gamma(\mu,\nu)}\!
    \int\!\norm{x-y}^2\,d\gamma(x,y)
    \Bigr)^{\!1/2}\!,
\end{equation}
where $\Gamma(\mu,\nu)$ is the set of all couplings---joint
distributions with marginals $\mu$ and~$\nu$.  A coupling $\gamma$
pairs each $x\sim\mu$ with a partner $y\sim\nu$, and $W_2$ is the
minimum root-mean-square displacement over all such pairings.

The operation central to our methodology is \emph{coupling
propagation}: given an initial coupling
$\gamma_0\in\Gamma(\mu_0,\nu_0)$, the transported coupling
$\gamma_t \triangleq (\phi_{t,0}\times\phi_{t,0})_\#\gamma_0$
flows both components forward under \eqref{eq:D_ode} and satisfies
$\gamma_t\in\Gamma(\mu_t,\nu_t)$.  Its cost
$\int\norm{x-y}^2 d\gamma_t$ upper-bounds $W_2^2(\mu_t,\nu_t)$.
The classical bound \eqref{eq:intro_classical} arises from bounding
this cost uniformly; our refinement bounds it in a
distribution-dependent manner.

\begin{definition}[$\lambda$-Contracting vector field]
\label{def:contraction}
$v(t,x)$ is $\lambda$-contracting ($\lambda>0$) if
\begin{equation}\label{eq:contraction_def}
    \inner{v(t,x)-v(t,y)}{x-y}\le -\lambda\norm{x-y}^2
    \quad\forall\;t,x,y.
\end{equation}
Equivalently, the symmetric part of the Jacobian satisfies
$\frac{1}{2}(\nabla_x v + \nabla_x v^\top)\preceq -\lambda I$.
\end{definition}

Condition~\eqref{eq:contraction_def} has a direct geometric meaning:
the component of the velocity difference along the line joining $x$
to~$y$ is directed inward with a magnitude of at least
$\lambda\norm{x-y}$.  In engineering terms, the system acts as a
distributed spring with stiffness $\ge\lambda$.  At the trajectory
level, $\norm{x_t-y_t}\le e^{-\lambda t}\norm{x_0-y_0}$; squaring,
integrating over an optimal coupling, and taking the square root
yields~\eqref{eq:intro_classical}.

However, this derivation replaces the \emph{local contraction rate}
$\lambda_{\mathrm{loc}}(t;x,y)\triangleq
-\inner{v(t,x)-v(t,y)}{x-y}/\norm{x-y}^2
\ge\lambda$
by its worst case~$\lambda$ before averaging.  For the cubic drift
$v(x)=-\lambda x-\beta x^3$, a direct calculation gives
\begin{equation}\label{eq:cubic_local}
    \lambda_{\mathrm{loc}}(x,y) = \lambda+\beta(x^2+xy+y^2),
\end{equation}
which exceeds $\lambda$ for all $(x,y)\neq(0,0)$ and grows
quadratically with distance from the origin.  At
$(x,y)=(-3,3)$, representative of initial distributions in our
experiments, the local rate is $\lambda+27\beta$, more than an order
of magnitude above the worst case when $\beta=1$.  This gap between
the worst-case rate and the typical rate under the distribution is the
conservatism in~\eqref{eq:intro_classical} that
Theorem~\ref{thm:eff_method} eliminates.

\subsection{Standing assumptions}\label{subsec:assumptions}

\begin{assumption}[Deterministic contraction]\label{ass:det_method}
For Regime~(D):
\textnormal{(D1)}~$\lambda$-contraction~\eqref{eq:contraction_def};
\textnormal{(D2)}~global Lipschitz: $\norm{v(t,x)-v(t,y)}\le
L\norm{x-y}$;
\textnormal{(D3)}~second-moment finiteness of $\mu_t$.
\end{assumption}

The ratio $\kappa=L/\lambda\ge 1$ acts as a condition number:
well-conditioned fields ($\kappa$ near $1$) admit favorable
discretization thresholds; stiff fields ($\kappa\gg 1$) require
smaller step sizes.

\begin{assumption}[Autonomous smooth contraction]\label{ass:euler_method}
For the Euler analysis, $v:\R^d\to\R^d$ is time-independent with:
\textnormal{(E1)}~unique equilibrium $x^\star$;
\textnormal{(E2)}~$\lambda$-contraction;
\textnormal{(E3)}~Lipschitz Jacobian: $\norm{\nabla v(x)}\le L$,
$\norm{\nabla v(x)-\nabla v(y)}\le M\norm{x-y}$.
\end{assumption}

The Lipschitz Jacobian condition~(E3) controls the local truncation
error through a Taylor expansion of the exact flow.  Together with
(E1), it ensures $\norm{v(x)}\le L\norm{x-x^\star}$, which decays
along exact trajectories at a rate $e^{-\lambda t}$.  This
decay of the drift and hence of the truncation error is the
physical mechanism behind the self-correcting behavior characterized
in Theorem~\ref{thm:selfcorr_method}.

\begin{assumption}[Ergodic additive-noise contraction]\label{ass:sde_method}
For Regime~(S):
\textnormal{(S1)}~Lipschitz $\lambda$-contraction;
\textnormal{(S2)}~$\sigma>0$ with isotropic diffusion;
\textnormal{(S3)}~full-support invariant measure $\pi_\sigma$.
\end{assumption}

Condition~(S3) is automatic under (S1) - (S2) and ensures the strict
inequality in Theorem~\ref{thm:sde_method}.

\section{Methodology and Validation}\label{sec:method}

We develop three results addressing gaps (G1)--(G3), each followed by
numerical validation.  One-dimensional Wasserstein distances are
computed exactly via sorting ($N=20{,}000$ particles);
two-dimensional experiments use the sliced-Wasserstein approximation
with 300 projections.

\subsection{Contribution~1: Effective contraction rate}\label{subsec:eff_method}

The classical derivation of~\eqref{eq:intro_classical} replaces the
local rate $\lambda_{\mathrm{loc}}\ge\lambda$ by its worst
case~$\lambda$ before integrating over the coupling. The inequality
$\E[f(X)]\ge(\inf f)\cdot 1$.  Our approach retains the full
distributional information by computing $\E[f(X)]$ directly.

\paragraph*{Effective rate.}
The coupling-propagated effective contraction rate is
\begin{equation}\label{eq:lambda_eff}
    \lambda_{\mathrm{eff}}(t;\gamma_t)\triangleq
    \frac{\displaystyle\int\lambda_{\mathrm{loc}}(t;x,y)\,
          \norm{x-y}^2\,d\gamma_t(x,y)}
         {\displaystyle\int\norm{x-y}^2\,d\gamma_t(x,y)},
\end{equation}
where $\gamma_t=(\phi_{t,0}\times\phi_{t,0})_\#\gamma_0$ is the
transported coupling and $\lambda_{\mathrm{loc}}$ is the local
contraction rate defined in~\eqref{eq:cubic_local} (see
Section~\ref{subsec:w2}).  The displacement weighting $\norm{x-y}^2$
arises naturally from differentiating $\int\norm{x_t-y_t}^2
d\gamma_0$ with respect to time; it amplifies the contribution of
distant pairs, which are precisely those that experience the strongest
nonlinear contraction.  As a result,
$\lambda_{\mathrm{eff}}(t;\gamma_t)$ can substantially exceed
$\lambda$ for non-affine fields.

\begin{theorem}[Effective-rate bound]\label{thm:eff_method}
Under Assumption~\ref{ass:det_method}, for any
$\gamma_0\in\Gamma(\mu_0,\nu_0)$ and transported
coupling~$\gamma_t$,
\begin{equation}\label{eq:eff_bound}
    W_2(\mu_t,\nu_t)
    \le \exp\!\Big(\!-\!\int_0^t\!
    \lambda_{\mathrm{eff}}(s;\gamma_s)\,ds\Big)
    \Big(\!\int\!\norm{x-y}^2\,d\gamma_0\Big)^{\!1/2}\!.
\end{equation}
Choosing $\gamma_0$ as an optimal $W_2$~coupling at time~$0$ gives
\begin{equation}\label{eq:eff_w2}
    W_2(\mu_t,\nu_t) \le
    \exp\!\Big(\!-\!\int_0^t\!
    \lambda_{\mathrm{eff}}(s;\gamma_s)\,ds\Big)\,
    W_2(\mu_0,\nu_0),
\end{equation}
which strictly refines~\eqref{eq:intro_classical} whenever $v$ is not
affine.
\end{theorem}

\begin{proof}
For paired trajectories $(x_t,y_t)$ under $\gamma_t$,
\[
    \frac{d}{dt}\norm{x_t-y_t}^2
    = 2\inner{v(t,x_t)-v(t,y_t)}{x_t-y_t}
    = -2\lambda_{\mathrm{loc}}(t;x_t,y_t)\norm{x_t-y_t}^2.
\]
Define $D(t)\triangleq\int\norm{x-y}^2 d\gamma_t$.  Differentiating
under the integral gives
$\dot D = -2\lambda_{\mathrm{eff}}(t;\gamma_t)\,D(t)$, whence
$D(t) = D(0)\exp\!\bigl(-2\int_0^t\lambda_{\mathrm{eff}}(s)\,ds
\bigr)$.
Since $W_2^2(\mu_t,\nu_t)\le D(t)$,
\eqref{eq:eff_bound} follows.  The strict refinement
over \eqref{eq:intro_classical} is immediate from
$\lambda_{\mathrm{eff}}\ge\lambda$, with equality only when
$\lambda_{\mathrm{loc}}=\lambda$ $\gamma_t$-a.e. -- i.e., only for
affine $v$.
\end{proof}

When $v(t,x)=-\lambda x-g(t,x)$ with monotone residual
$\inner{g(x)-g(y)}{x-y}\ge 0$, the rate decomposes as
\begin{equation}\label{eq:mono_refine}
    \lambda_{\mathrm{eff}}(t;\gamma_t) \ge \lambda +
    \underbrace{
    \frac{\int\inner{g(x)-g(y)}{x-y}\,d\gamma_t}
         {\int\norm{x-y}^2\,d\gamma_t}
    }_{\triangleq\;\Delta\lambda(t;\gamma_t)\;\ge\;0}\!.
\end{equation}
For $g(x)=\beta x^3$, this excess is positive whenever the
distributions have support away from the origin.

\begin{remark}[On optimality]\label{rem:not_optimal}
The transported coupling is generally not optimal for
$W_2(\mu_t,\nu_t)$, so~\eqref{eq:eff_w2} is an upper bound.
In one dimension, the sorted coupling remains optimal under monotone
flows, and the bound becomes an identity.
\end{remark}

\subsubsection*{Validation (Fig.~\ref{fig:eff_rate})}
The cubic drift $v(x)=-x-\beta x^3$, $\lambda=1$, is simulated with
$N=20{,}000$ sort-matched particles, $\mu_0=\N(-3,1)$,
$\nu_0=\N(3,1.5^2)$, $T=5$, $h=0.001$.

Panel~(a) shows $\hat\lambda_{\mathrm{eff}}(t;\gamma_t)$ for
$\beta\in\{0,0.2,0.5,1.0\}$.  At $t=0$ the distributions are
centred at $\pm 3$ where the cubic term dominates, yielding
initial rates of $\hat\lambda_{\mathrm{eff}}(0)\approx 3$
($\beta=0.2$) up to $\approx 15$ ($\beta=1.0$).  As contraction
concentrates mass toward zero, the local rate surplus diminishes and
$\hat\lambda_{\mathrm{eff}}(t)$ decays monotonically
toward~$\lambda$.  The $\beta=0$ (linear) curve stays flat at~$1$
throughout, confirming the equality case of
Theorem~\ref{thm:eff_method}.  The shaded region under the
$\beta=0.5$ curve visualises the accumulated excess
$\int_0^t\Delta\lambda(s)\,ds$ that drives the tighter bound.

Panel~(b) compares three quantities for $\beta=0.5$: the actual
$W_2(\mu_t,\nu_t)$, the effective-rate
bound~\eqref{eq:eff_w2} (green dashed), and the classical bound
(red dashed).  The effective-rate bound tracks the actual $W_2$
closely, as they nearly coincide in 1-D, where the sorted coupling
remains optimal under monotone flows
(Remark~\ref{rem:not_optimal}), while the classical bound is
substantially loose, with the shaded gap quantifying the improvement
from Theorem~\ref{thm:eff_method}.

As a calibration baseline, affine systems ($v=-\lambda x$) were also
tested for $\lambda\in\{0.5,1,2,4\}$ and confirmed the exact tightness of
the classical bound, verifying that the improvement is entirely
attributable to the nonlinear excess $\Delta\lambda$.
Additionally, sweeping $\beta\in\{0,0.2,0.5,1.0\}$ shows the gap
grows monotonically with $\beta$, consistent
with \eqref{eq:mono_refine}.

\smallskip\noindent\textit{Parametric validation
(Fig.~\ref{fig:nonlinear}).}\;
To isolate the effect of nonlinearity, Fig.~\ref{fig:nonlinear}
presents a dedicated study of the cubic drift with
$\mu_0=\N(-3,1)$, $\nu_0=\N(3,2.25)$.  Panel~(a) shows the actual
$W_2$ (solid) versus the classical bound (dashed) for $\beta=0.5$;
the shaded region is the gap that the classical bound
overestimates.  The Euler step $h=0.002$ contributes a discretization
error of $\Od(h)\approx 0.002$, negligible relative to the bound gap
of order~$1$.  Panel~(b) sweeps $\beta\in\{0,0.2,0.5,1.0\}$: the
gap grows monotonically with~$\beta$, and at $\beta=0$ the actual and
classical curves coincide, confirming that the entire gap is generated
by the nonlinear excess $\Delta\lambda$ in~\eqref{eq:mono_refine}.

\begin{figure*}[t]
\centering
\includegraphics[width=\textwidth]{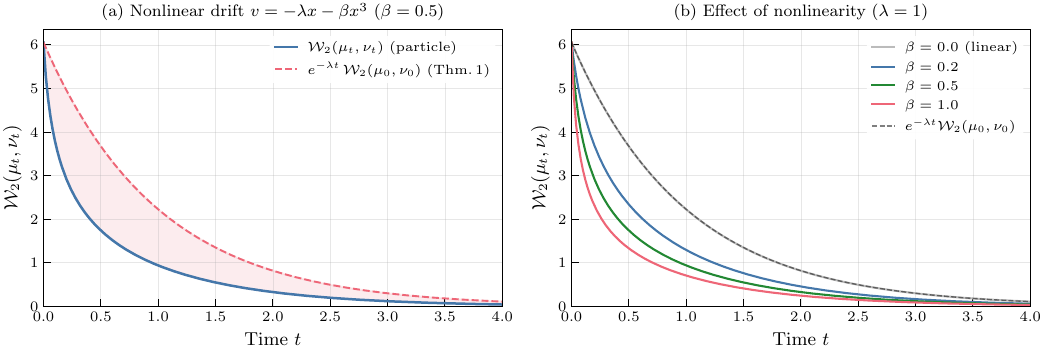}
\caption{Nonlinear bound gap for $v(x)=-x-\beta x^3$, $\lambda=1$.
(a)~Particle-computed $W_2(\mu_t,\nu_t)$ (solid) for $\beta=0.5$
versus $e^{-\lambda t}W_2(\mu_0,\nu_0)$ (dashed); the shaded region
is the conservatism of the classical bound.
(b)~$\beta$-sweep: the gap grows monotonically; $\beta=0$ recovers
tightness.}
\label{fig:nonlinear}
\end{figure*}

\begin{figure*}[t]
\centering
\includegraphics[width=\textwidth]{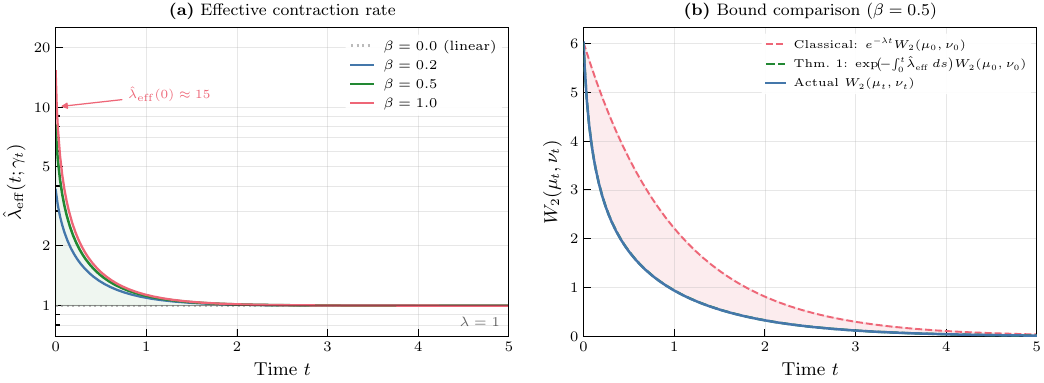}
\caption{Effective contraction rate for $v(x)=-x-\beta x^3$,
$\lambda=1$.
(a)~$\hat\lambda_{\mathrm{eff}}(t;\gamma_t)$ for
$\beta\in\{0,0.2,0.5,1.0\}$: starts well above $\lambda$ and decays
as mass concentrates near the origin.
(b)~Three-way comparison for $\beta=0.5$: the effective-rate
bound (green dashed) captures nearly all the improvement over the
classical bound (red dashed); the shaded region is the gap closed by
Theorem~\ref{thm:eff_method}.}
\label{fig:eff_rate}
\end{figure*}

\subsection{Contribution~2: Self-correcting Euler
discretization}\label{subsec:euler_method}

The explicit Euler scheme for the autonomous ODE $\dot x=v(x)$ with
a unique equilibrium $x^\star$ is
$\tilde x_{k+1}=\tilde x_k+h\,v(\tilde x_k)$, $t_k=kh$.
Let $\mu_{t_k}$ denote the exact law and $\tilde\mu_k$ the Euler law,
both starting from $\mu_0$. The Wasserstein discretization error is
$e_k\triangleq W_2(\mu_{t_k},\tilde\mu_k)$.

In a standard ODE solver, the global error is governed by two
competing forces: local truncation error (injecting new error at each
step) and any damping in the dynamics (dissipating accumulated error).
For generic systems, the truncation forcing is roughly constant, and
the error rises monotonically to a plateau.

For contracting flows, a fundamentally different behavior emerges.
The drift vanishes at equilibrium ($v(x^\star)=0$) and satisfies
$\norm{v(x)}\le L\norm{x-x^\star}$, so the drift magnitude along
exact trajectories inherits the exponential decay $e^{-\lambda t}$.
The truncation forcing therefore \emph{shrinks exponentially},
producing a competition: errors accumulate on scale~$\Od(1)$ while
contraction damps them on scale $\Od(1/\lambda)$.  The result is a
non-monotone error profile peaking at the crossover.

\begin{lemma}[Euler contractivity]\label{lem:euler_method}
Under Assumption~\ref{ass:euler_method}\textnormal{(E2)--(E3)},
$\Psi_h(x)=x+hv(x)$ satisfies
$\norm{\Psi_h(x)-\Psi_h(y)}^2
\le (1-2\lambda h + L^2 h^2)\norm{x-y}^2$.
The contraction factor $q(h)=\sqrt{1-2\lambda h+L^2 h^2}<1$
when $h<2\lambda/L^2$.
\end{lemma}

\begin{theorem}[Self-correcting envelope]\label{thm:selfcorr_method}
Under Assumption~\ref{ass:euler_method} with $h<2\lambda/L^2$, let
$C_0 = \frac{L^2}{2}e^{Lh}(\E\norm{X_0-x^\star}^2)^{1/2}$.  Then
\begin{equation}\label{eq:envelope}
    \boxed{\;e_k \le C_0\,h\,t_k\,e^{-\lambda t_k}\;}
    \qquad \forall\;k\ge 0.
\end{equation}
This implies: \textnormal{(i)}~peak time $t^\star=1/\lambda$,
independent of all other parameters;
\textnormal{(ii)}~peak magnitude $\Od(h/\lambda)$;
\textnormal{(iii)}~terminal-to-peak ratio
$e(\lambda T)e^{-\lambda T}$, exponentially small for
$\lambda T\gg 1$.
\end{theorem}

\begin{proof}[Proof sketch]
Taylor expansion gives local truncation
$\norm{x(t_{k+1})-\Psi_h(x(t_k))} \le
\frac{L^2}{2}e^{Lh}h^2\norm{x(t_k)-x^\star}$.  Since
$\norm{x(t_k)-x^\star}\le e^{-\lambda t_k}\norm{x_0-x^\star}$,
forcing at step~$k$ is $\Od(h^2 e^{-\lambda t_k})$.  Combining with
$q(h)<1$ from Lemma~\ref{lem:euler_method} and summing
yields~\eqref{eq:envelope}.  The peak follows from
$(te^{-\lambda t})'=0$ at $t^\star=1/\lambda$.
\end{proof}

The envelope $C_0 h t e^{-\lambda t}$ encodes the full non-monotone
structure: the factor~$t$ reflects cumulative truncation errors, while
$e^{-\lambda t}$ reflects both decaying forcing and contraction-driven
damping.  Their product rises until $t^\star=1/\lambda$---where
accumulation and damping balance---then falls as contraction dominates,
producing a self-correcting error trajectory from first principles.
The stability threshold $h_{\max}=2\lambda/L^2=2/(\kappa L)$ is more
permissive than the standard CFL condition $h<1/L$ whenever
$\kappa<2$, an improvement specific to contraction.

\begin{remark}[Failure without contraction]\label{rem:contrast}
Without contraction ($\lambda=0$), the error grows as
$\Od(ht)$ without bound.  For volume-preserving rotations
$v=[-x_2,x_1]^\top$, Euler amplifies distances by $\sqrt{1+h^2}>1$
per step, producing monotonically increasing error.  Self-correction
is a direct and specific consequence of contraction.
\end{remark}

\subsubsection*{Validation (Fig.~\ref{fig:ablation})}
Four 2-D fields are compared: strongly contracting ($v=-2x$,
$\lambda=2$), weakly contracting ($v=-0.5x$, $\lambda=0.5$),
volume-preserving rotation ($v=[-x_2,x_1]^\top$, $\lambda=0$), and
expanding ($v=0.5x$, $\lambda<0$);
$\mu_0=\N([-2,0]^\top,I)$, $\nu_0=\N([2,0]^\top,I)$.

Panel~(a) confirms distributional contraction: the contracting flows
decay exponentially at their respective rates, the rotation preserves
$\mathrm{SW}_2$ (isometry), and the expanding flow diverges.

Panel~(b) shows the Euler discretization error ($h=0.08$) measured
against a fine-step reference ($h=0.001$).  The $\lambda=2$ error
curve peaks near $t\approx 0.5=1/\lambda$ and the $\lambda=0.5$ curve
near $t\approx 2=1/\lambda$,
matching Theorem~\ref{thm:selfcorr_method}(i) exactly.  After the peak, both
curves decay exponentially, confirming the self-correcting envelope.
The rotation and expanding flows show monotonically increasing error
throughout, precisely as predicted by Remark~\ref{rem:contrast}.  This
ablation demonstrates that self-correction is a consequence of the
contraction structure, not of the discretization scheme: removing
contraction removes the self-correction.

Terminal error versus step size (not shown) confirms $\Od(h)$ scaling
for all contracting cases, consistent with
Theorem~\ref{thm:selfcorr_method}(ii).

\begin{figure*}[t]
\centering
\includegraphics[width=\textwidth]{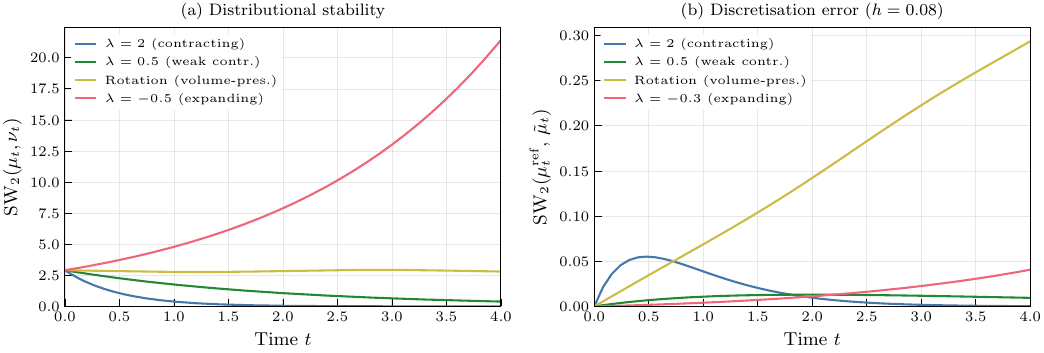}
\caption{Contracting vs.\ non-contracting ablation in 2-D.
(a)~$\mathrm{SW}_2(\mu_t,\nu_t)$: only contracting flows decay.
(b)~Euler error: contracting flows produce self-correcting profiles
peaking at $t^\star=1/\lambda$; rotation and expanding flows
accumulate error monotonically.}
\label{fig:ablation}
\end{figure*}

\subsection{Contribution~3: Nonlinear stochastic
robustness}\label{subsec:sde_method}

Our final result addresses gap~(G3): \emph{how tightly does a
contracting flow confine its invariant measure around equilibrium
under persistent noise?}

For the linear OU process $dX=-\lambda X\dt+\sigma\dW$, the stationary
variance is exactly $d\sigma^2/(2\lambda)$, reflecting a precise
balance: contraction removes $2\lambda\E\norm{X}^2$ units of variance
per unit time and diffusion injects $d\sigma^2$.  For nonlinear
$\lambda$-contracting drifts, a Lyapunov argument provides the same result as
an \emph{inequality}.  But a superlinear drift exerts
disproportionately strong restoring forces on large displacements;
Theorem~\ref{thm:sde_method} proves the inequality is strict for
every non-affine contracting drift.

Define the radial rate
$\lambda_{\mathrm{rad}}(x) =
-\inner{v(x)}{x-x^\star}/\norm{x-x^\star}^2 \ge\lambda$
(with $\lambda_{\mathrm{rad}}(x^\star)=\lambda$)
and the stationary effective rate
\begin{equation}\label{eq:lambda_eff_inf}
    \lambda_{\mathrm{eff}}^\infty \triangleq
    \frac{\E_{\pi_\sigma}\!
      [\lambda_{\mathrm{rad}}(X)\norm{X-x^\star}^2]}
         {\E_{\pi_\sigma}\!\norm{X-x^\star}^2}
    \ge\lambda.
\end{equation}

\begin{theorem}[Stationary variance and strict improvement]
\label{thm:sde_method}
Under Assumption~\ref{ass:sde_method}:
\textnormal{(i)}
$\E_{\pi_\sigma}\norm{X-x^\star}^2
= d\sigma^2/(2\lambda_{\mathrm{eff}}^\infty)
\le d\sigma^2/(2\lambda)$.
\textnormal{(ii)}~If $v\in C^1$ is not affine, then
$\lambda_{\mathrm{eff}}^\infty>\lambda$ and the inequality is strict.
\textnormal{(iii)}
$W_2(\pi_\sigma,\delta_{x^\star})
\le \sigma\sqrt{d/(2\lambda_{\mathrm{eff}}^\infty)}$,
strict for non-affine~$v$.
\end{theorem}

\begin{proof}
It\^o's formula on $V(t)=\E\norm{X_t-x^\star}^2$ gives
$\dot V = -2\E[\lambda_{\mathrm{rad}}(X_t)\norm{X_t-x^\star}^2]
+ d\sigma^2$.
At stationarity $\dot V=0$, yielding~(i).  For~(ii):
$\lambda_{\mathrm{rad}}$ is continuous and exceeds~$\lambda$ on a
nonempty open set~$U$; by~(S3), $\pi_\sigma(U)>0$, so
$\lambda_{\mathrm{eff}}^\infty>\lambda$.
Part~(iii) follows from
$W_2^2(\pi_\sigma,\delta_{x^\star})\le
\E_{\pi_\sigma}\norm{X-x^\star}^2$.
\end{proof}

\begin{remark}[Hardening-spring analogy]\label{rem:spring}
A linear spring of stiffness~$\lambda$ under white noise reaches
amplitude $\sigma^2/\lambda$.  A \emph{hardening spring}---whose
stiffness increases with displacement---resists large excursions more
effectively, producing a tighter distribution.
\end{remark}

\subsubsection*{Validation (Fig.~\ref{fig:stochastic})}
The SDE $dX=(-X-0.5X^3)\dt+\sigma\dW$ with $\lambda=1$;
Euler--Maruyama, $h=0.005$, $N=12{,}000$.

Panel~(a): the empirical $\mathrm{Var}(X_t)$ converges to a
stationary value strictly below $\sigma^2/(2\lambda)$ (dotted) for
all $\sigma\in\{0.3,0.5,0.8,1.2\}$, confirming
Theorem~\ref{thm:sde_method}(ii).  The gap grows with $\sigma$
because larger noise drives the system to larger displacements where
$\lambda_{\mathrm{rad}}(x)=1+0.5x^2$ exceeds the baseline most
strongly, amplifying the hardening-spring effect of
Remark~\ref{rem:spring}.

Panel~(b): $W_2(\pi_\sigma,\delta_{x^\star})$ at stationarity
($T=18$) lies consistently below $\sigma/\!\sqrt{2\lambda}$ over
$\sigma\in[0.05,2]$ with a growing gap, validating
Theorem~\ref{thm:sde_method}(iii). The gap quantifies the
noise-rejection improvement that the nonlinear contracting drift
provides over a linear spring of the same minimum stiffness.

\begin{strip}
\centering
\includegraphics[width=\textwidth]{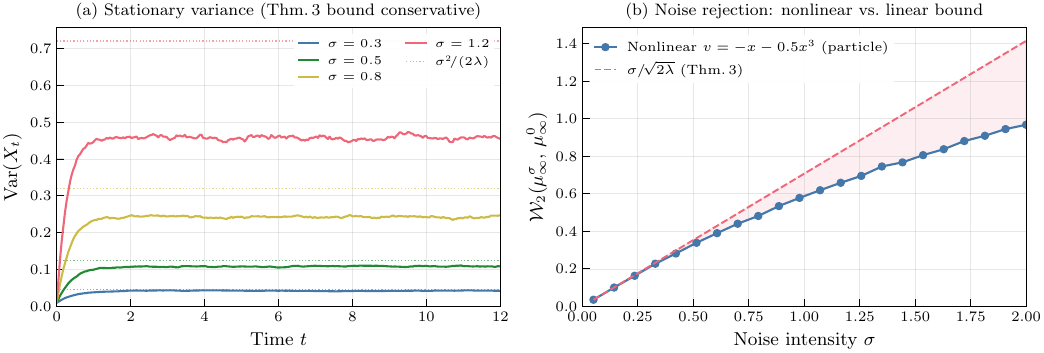}
\captionof{figure}{Stochastic robustness for
$dX=(-X-0.5X^3)\dt+\sigma\dW$, $\lambda=1$.
(a)~$\mathrm{Var}(X_t)$ (solid) versus $\sigma^2/(2\lambda)$
(dotted): stationary values lie strictly below the linear bound.
(b)~$W_2(\pi_\sigma,\delta_{x^\star})$ (blue) versus
$\sigma/\!\sqrt{2\lambda}$ (red dashed): the nonlinear drift rejects
noise more effectively.}
\label{fig:stochastic}
\end{strip}

\section{Conclusion}\label{sec:conclusion}

We have extended the distributional stability theory of contracting
flows beyond the affine regime.  The coupling-propagated effective
rate (Theorem~\ref{thm:eff_method}) captures nonlinear contraction
missed by the classical bound; the self-correcting envelope
(Theorem~\ref{thm:selfcorr_method}) provides explicit peak-time and
decay formulae for Euler discretisation; and the nonlinear stochastic
bound (Theorem~\ref{thm:sde_method}) is provably strict for non-affine
drifts.  These results jointly support contraction-regularized
flow-matching~\cite{lipman2023flow}: Theorem~\ref{thm:selfcorr_method}
bounds inference discretisation error, Theorem~\ref{thm:sde_method}
bounds robustness to training noise, and
Theorem~\ref{thm:eff_method} predicts that practical quality exceeds
worst-case guaranties for nonlinear architectures.  Future work
includes extension to implicit schemes and incorporation of
finite-sample estimation error.

\bibliographystyle{IEEEtran}
\bibliography{ref}

\end{document}